\documentclass{article}[11pt]
\usepackage{graphicx} 
\usepackage{hyperref}
\usepackage{xcolor}
\usepackage{amsthm}
\usepackage{tikz}
\usepackage{amsmath}
\usepackage{amssymb}
\newcommand{\hsp}{\hspace{.5mm}}
\newcommand{\ubar}{\underline{u}}
\newcommand{\Cbar}{\underline{C}}
\newcommand{\hba}{\underline{h}}
\newcommand{\be}{\begin{equation}}

\newcommand{\hb}{\underline{h}}
\newcommand{\ee}{\end{equation}}
\newcommand{\inta}{\int_{v_1}^{v_a}}
\newcommand{\dv}{\hspace{1mm}\text{d}v}
\newcommand{\dvprime}{\hspace{1mm} \text{d}v^{\prime}}
\newcommand{\upr}{u^{\prime}}

\newcommand\restri[2]{{
		\left.\kern-\nulldelimiterspace 
		#1 
		\right|_{#2} 
}}

\newtheorem{proposition}{Proposition}[section]

\newtheorem{theorem}{Theorem}[section]

\newtheorem{lemma}{Lemma}[section]

\numberwithin{equation}{section}
\allowdisplaybreaks[4]

\title{Formation of trapped surfaces for the spherically symmetric Einstein-Yang-Mills system with non-trivial incoming data }
\author{Nikolaos Athanasiou, Puskar Mondal, Shing-Tung Yau}
\date{August 2026}

\begin{document}

\maketitle
\begin{abstract}
\noindent We establish a trapped-surface formation theorem for the spherically symmetric Einstein--Yang--Mills equations in a double-null gauge. The theorem concerns characteristic initial data posed on a pair of transversely intersecting null hypersurfaces and allows non-trivial incoming data. The proof extends the singular characteristic method of An-Lim \cite{AnLim} for the Einstein--Maxwell--scalar field system to the non-abelian Yang--Mills setting, where the curvature coupling and gauge-field nonlinearities introduce new structural difficulties. This paper constitutes the first part of a program toward weak cosmic censorship for the Einstein--Yang--Mills system.  
\end{abstract}

\section{Introduction}

\indent The Einstein--Yang--Mills equations provide one of the simplest geometric
models in which the nonlinear interaction between gravity and a genuinely
non-Abelian gauge field can be studied beyond perturbative regimes. The behavior of this system in spherical symmetry is substantially richer than that of
the Einstein--Maxwell system. Indeed, spherically symmetric electrovacuum
solutions are locally contained in the Reissner--Nordstr\"om family, by the
electrovacuum version of the celebrated Birkhoff's theorem.  In particular,
spherical symmetry eliminates all radiative Maxwell degrees of freedom.
For the Einstein--Yang--Mills system, by contrast, the purely magnetic
\(SU(2)\) ansatz retains a nonlinear scalar degree of freedom \(w\), and
the reduced equations remain a genuinely coupled system of wave and
transport equations.

This distinction in rigidity is already visible in the static theory. Bartnik and
McKinnon numerically discovered a countable family of globally regular,
particle-like solutions \cite{bartnik1988particlelike}, whose existence was
subsequently established rigorously by Smoller, Wasserman, Yau, and McLeod
\cite{smoller1991smooth}. Smoller, Wasserman, and Yau later proved that, for
every prescribed horizon radius, there exists an infinite family of static
\(SU(2)\) Einstein--Yang--Mills black holes with a regular event horizon
\cite{smoller1993existence}. These solutions are indexed by a winding
number, have finite ADM mass, and approach a Schwarzschild metric at
spatial infinity at rate \(O(r^{-2})\), despite carrying nontrivial
Yang--Mills fields in the interior. Thus the geometry at infinity does not
determine the non-Abelian field in the exterior region. The corresponding
solutions are nevertheless unstable \cite{straumann1990instability}, which
raises the dynamical question underlying the present work:

\begin{quote}
Under what quantitative conditions can regular characteristic data for the
spherically symmetric Einstein--Yang--Mills system evolve to form a trapped
surface?
\end{quote}

\noindent The purpose of this article is to establish a trapped-surface formation
criterion for the purely magnetic \(SU(2)\) Einstein--Yang--Mills system.
The initial data are prescribed on two intersecting null hypersurfaces. A
principal feature of our result is that the data on the incoming initial
hypersurface are allowed to be geometrically and gauge-theoretically
nontrivial. Thus the argument does not rely on the presence of an initially
Minkowskian incoming region. The theorem is formulated entirely in terms of
the Hawking mass, the relative radial width of the initial pulse, and the
non-Abelian magnetic potential. For details on the local existence theorem for this system, we refer the reader to \cite{Li_Wang}.

\subsection{The spherically symmetric reduction in a double null gauge}

Let \((\mathcal M,g)\) be a connected, time-oriented, globally hyperbolic
four-dimensional Lorentzian manifold admitting an effective isometric
\(SO(3)\)-action whose principal orbits are spacelike two-spheres. We denote
by \(\Gamma\) the set of fixed points of the action, when it is present, and
refer to \(\Gamma\) as the axis of symmetry. On the two-dimensional quotient
\(\mathcal Q=\mathcal M/SO(3)\) we introduce double-null coordinates
\((u,v)\). The construction is as follows:

\begin{enumerate}
\item We fix a point $q$ on the Penrose diagram and we label the outgoing null geodesic intersecting $q$ by $C$ and the incoming null geodesic intersecting $q$ by $\Cbar$. Given that the Penrose diagram suppresses two angular dimensions, in the actual spacetime $C$ and $\Cbar$ are null hypersurfaces.
 \item We parametrize $C$ and $\Cbar$ by the variables $v, u$ respectively and at the intersection of $\Gamma$ and $\Cbar$, we normalize $u=0$. Similarly, at the intersection of $\Gamma$ and $C$ we set $v=0$. Let the coordinates of $q$ be $q(u_0,v_1)$. 

 \item Once coordinates have been fixed on $C$ and $\Cbar$ we can define a \textit{double-null coordinate} system on the domain of dependence of $C\cup \Cbar$ in a standard way, by looking at any given fixed point $p$ in this domain and tracing the point back through an incoming and an outgoing null geodesic to $C$ and $\Cbar$ respectively. If the coordinates where the geodesics passing through $p$ intersect $\Cbar$ and $C$ respectively are $u$ and $v$, we define $p:=p(u,v)$. Let $\mathcal{D}(0,v_1)$ denote the spacetime region that is bounded by $C,\Cbar$ and $\Gamma$.
\end{enumerate}

\noindent The corresponding null hypersurfaces are
\[
C_u:=\{u=\mathrm{constant}\},
\qquad
\underline C_v:=\{v=\mathrm{constant}\},
\]
and their intersections
\[
S_{u,v}:=C_u\cap\underline C_v
\]
are the symmetry spheres. The area radius \(r\) is defined invariantly by
\begin{equation}
    \operatorname{Area}(S_{u,v})=4\pi r^2(u,v).
\end{equation}
In a double-null gauge the metric takes the form
\begin{equation}
    g=-4\Omega^2(u,v)\,\mathrm du\,\mathrm dv
      +r^2(u,v)\,d\sigma_{\mathbb S^2},
    \label{eq:intro-double-null-metric}
\end{equation}
where \(\Omega>0\) is the lapse function and \(d\sigma_{\mathbb S^2}\) is the standard round
metric on \(\mathbb S^2\).

\vspace{3mm}

 \noindent We are concerned with the dynamics of a Yang-Mills field with gauge group $SU(2)$, the real Lie group of unitary matrices of determinant $1$. 
\noindent Let \(\{\tau_1,\tau_2,\tau_3\}\) be a basis of of the associated Lie algebra 
\(\mathfrak{su}(2)\) (which can be identified with the set of all anti-hermitian, traceless $2\times 2$ matrices) satisfying
\[
    [\tau_i,\tau_j]=\varepsilon_{ijk}\tau_k.
\]Here, $\varepsilon_{ijk}$ for $i,j,k \in \{1,2,3\}$ is the Levi-Civita symbol. For example, we can set

\begin{equation}
\tau_1 := \frac{i}{2}\begin{pmatrix}
0 & 1 \\ 1 & 0
\end{pmatrix},\hspace{1mm} \tau_2 := \frac{1}{2}\begin{pmatrix}
0 & -1 \\ 1 & 0
\end{pmatrix}, \hspace{1mm} \tau_3 := \frac{i}{2}\begin{pmatrix}
1 & 0 \\ 0 & -1
\end{pmatrix}.
\end{equation}
We henceforth impose the spherically symmetric \textit{purely magnetic} ansatz, where the  Yang-Mills connection $1-$form  $A$ takes the following form: 
\begin{equation}
    A
    =
    w\,\tau_1\,\mathrm d\theta
    +
    \bigl(\cot\theta\,\tau_3+w\,\tau_2\bigr)
       \sin\theta\,\mathrm d\phi ,
    \label{eq:intro-magnetic-ansatz}
\end{equation}
where \(w=w(u,v)\) and $\theta, \phi$ are the standard angular coordinates on $S^2$. Its curvature is
\begin{equation}
\begin{split}
    F
    = dA + A \wedge A =
    \mathrm dw\wedge
    (
       \tau_1\,\mathrm d\theta
       +\sin\theta\,\tau_2\,\mathrm d\phi
    ) 
        -(1-w^2)\tau_3\sin\theta\,
       \mathrm d\theta\wedge\mathrm d\phi .
\end{split}
\label{eq:intro-curvature}
\end{equation}
The reason \eqref{eq:intro-magnetic-ansatz} is thus named is that $F_{uv}=0$, meaning that the electric charge vanishes.

\vspace{3mm}

\noindent We set
\begin{equation}
    Q:=w^2-1.
    \label{eq:intro-Q}
\end{equation}
Be it noted from the start that $Q$ should not be thought of as a charge per se. It is, nevertheless, not without physical meaning. Within the purely magnetic reduction, \(Q/r^2\) is the coefficient of the
angular curvature, in the sense that
\[ F_{AB} = \frac{Q}{r^2}\tau_3 \varepsilon_{AB},   \]where the capital Latin letters $\{A,B 
\}$ refer to indices on the $2-$sphere $S^2$. We shall refer to \(Q\) as the magnetic
charge aspect, following the terminology used in the spherical
Einstein--Yang--Mills literature. It is important, however, to note   that \(Q\) is
neither a conserved Gauss charge nor a gauge-invariant Yang--Mills charge
in the unrestricted theory. In particular,
\begin{equation}
    \partial_a Q=2w\,\partial_a w,
    \qquad a\in\{u,v\},
\end{equation}
implying that $Q$ has nontrivial dynamics. We shall see that this quantity will be important for us later on. \vspace{3mm}

\noindent The Spherically Symmetric Einstein-Yang-Mills system (SSEYM) takes the form

\vspace{3mm}

\begin{equation} \label{eq:Einsteineq}
R_{\mu\nu}-\frac{1}{2} R \hsp g_{\mu\nu} = 2 T_{\mu\nu},
\end{equation}
\begin{equation}\label{eq:YMeq}
\hat{D}^{\mu}F_{\mu\nu}=0.
\end{equation}Here $\hat{D}_{\alpha}:= D_{\alpha}+[A_\alpha,\cdot]$ denotes the gauge-covariant derivative, $[\cdot,\cdot]$ is the Lie bracket on $\mathfrak{s}\mathfrak{u}(2)$ and \[ T_{\mu\nu} = \langle {F_{\mu}}^{\alpha}, F_{\nu \alpha}\rangle - \frac{1}{4}g_{\mu\nu}\langle F^{\rho\sigma},F_{\rho\sigma}\rangle.\]

\noindent Define
\begin{equation}
    h:=\Omega^{-2}\partial_v r,
    \qquad
    \underline h:=\partial_u r.
    \label{eq:intro-null-expansions}
\end{equation}
The Hawking mass \(m\) and the mass ratio \(\mu\) are given by
\begin{equation}
    m
    :=
    \frac r2\bigl(1+h\underline h\bigr),
    \qquad
    \mu:=\frac{2m}{r}.
    \label{eq:intro-Hawking-mass}
\end{equation}
Equivalently,
\begin{equation}
    1-\mu=-h\underline h
          =g^{\alpha\beta}
             \partial_\alpha r\,\partial_\beta r.
\end{equation}
In a regular, nontrapped region we orient the null coordinates so that
\[
    h>0,
    \qquad
    \underline h<0.
\]
A symmetry sphere \(S_{u,v}\) is \textit{trapped} precisely when
\[
    \partial_u r(u,v)<0,
    \qquad
    \partial_v r(u,v)<0.
\]

\noindent Under the ansatz \eqref{eq:intro-magnetic-ansatz}, equations
\eqref{eq:Einsteineq}-\eqref{eq:YMeq} reduce to
\begin{align}
    \partial_v h
    =
    -2\Omega^{-2}
       \frac{(\partial_v w)^2}{r},
    \label{eq:intro-Raychaudhuri-v}
    \end{align}
    \begin{align}
    \partial_u\bigl(\Omega^{-2}\underline h\bigr)
    =
    -2\Omega^{-2}
       \frac{(\partial_u w)^2}{r},
    \label{eq:intro-Raychaudhuri-u}
    \end{align}
    \begin{align} \partial_v\underline h
    =
    \partial_u(\Omega^2h)
    =
    \frac{\Omega^2}{r}
    \left(
        -\mu+\frac{Q^2}{r^2}
    \right),
    \label{eq:intro-cross-focusing}
    \end{align}
    \begin{align}
    \partial_u \partial_v \Omega = \partial_v \partial_u \Omega = \Omega^2 \frac{\mu}{r^2}- 2 \hsp \Omega^2 \frac{Q^2}{r^4}, \label{eq:intro-Omegaequation} \end{align}
    \begin{align}
    \partial_u\partial_v w
    +
    \Omega^2\frac{Q\hsp w}{r^2}
    =0.
    \label{eq:intro-YM-wave}
\end{align}
The Hawking mass satisfies the transport identities
\begin{align}
    \partial_v m
    &=
    \Omega^2h\,\frac{Q^2}{2r^2}
    -
    \Omega^{-2}\underline h\,(\partial_vw)^2,
    \label{eq:intro-mass-v}
    \\
    \partial_u m
    &=
    \underline h\,\frac{Q^2}{2r^2}
    -
    h\,(\partial_uw)^2.
    \label{eq:intro-mass-u}
\end{align}
Thus, as long as \(h>0\) and \(\underline h< 0\), the Hawking mass is
nondecreasing in the outgoing direction and nonincreasing in the incoming
direction. 

\vspace{3mm}

\par \noindent A characteristic initial data set for the SSEYM system consists of a triple $(r,\Omega, \partial_r w)$ prescribed on initial outgoing and incoming hypersurfaces $C, \underline{C}$ respectively, satisfying appropriate compatibility conditions on the initial sphere of intersection,  that further satisfies the following boundary conditions on the axis $\Gamma$:

\[ \restri{r}{\Gamma}=0,\restri{m}{\Gamma}=0 , \restri{Q}{\Gamma}=0    \]

There are two types of initial conditions to be considered for the given system:

\vspace{3mm}

\begin{itemize}
\item The first type of initial conditions concerns the purely geometric part, so to speak. In other words, this set of initial conditions is expected to be the same for all physical models. First of all, on the axis $\Gamma$ there must hold $r=0$. Moreover, a byproduct of spherical symmetry is that as we consider points infinitesimally close to the center, its incoming geodesics, as also explained in \cite{AnLim}, essentially become outgoing, which forces the compatibility assumption $\partial_v r(u_0,0)= -\partial_u r(u_0,0)$. Moreover, on $C\cup \Cbar$ we choose the normalization $\Omega^2 = \frac{1}{4}$, in coherence with \cite{Li_Wang}. This fixes the lapse initially and gets rid of a redundant degree of freedom.

\item The second type of initial conditions is matter dependent. Here, $\restri{Q}{\Gamma}=0$ (equivalently $w=1$ on $\Gamma$) is essential to ensuring the regularity of the Yang-Mills equation on the axis (see also \cite{Li_Wang}). As far as the free Yang-Mills parameter is concerned, we can prescribe $w$ freely on $C \cup \Cbar$. This completely determines the first derivatives of $w$ on $C$ and $\Cbar$. The evolution of $w$ on $C$ and $\Cbar$ is then determined by equation \eqref{eq:intro-YM-wave} (where $\Omega^2=\frac{1}{4}$ because of the lapse fixing condition we imposed).
\end{itemize}

\noindent For a more comprehensive explanation of these variables we refer the reader to \cite{Li_Wang}, from which the equations are taken. 
\noindent Equations \eqref{eq:intro-Raychaudhuri-v}--\eqref{eq:intro-mass-u}
display the central structural feature of the problem. The derivative
terms of \(w\) enter the Raychaudhuri equations with a favorable sign and
therefore enhance focusing. The magnetic term \(Q^2/r^3\), on the
other hand, enters the cross-focusing equation
\eqref{eq:intro-cross-focusing} with the opposite sign. Consequently,
Yang--Mills energy contributes positively to the Hawking mass while part
of the same energy simultaneously opposes gravitational focusing. This
competition has no analogue in the uncharged Einstein--scalar field
system.

\subsection{Statement of the main result}

Let
\[
    C
    :=
    \{u=u_0,\ v_1\leq v\leq v_2\},
    \qquad
    \underline C
    :=
    \{v=v_1,\ u_0\leq u\leq u_*\}
\]
be two intersecting initial null hypersurfaces. We consider smooth
characteristic data satisfying the constraint equations, the standard
orientation conditions
\[
    \partial_vr>0,
    \qquad
    \partial_ur<0,
\]
and the absence of initially trapped symmetry spheres. Set 
\begin{equation}
    r_1:=r(u_0,v_1),
    \qquad
    r_2:=r(u_0,v_2),
\end{equation}
and define the initial relative radial width and normalized mass
concentration by
\begin{equation}
    \delta_0
    :=
    \frac{r_2-r_1}{r_2},
    \qquad
    \eta_0
    :=
    \frac{m(u_0,v_2)-m(u_0,v_1)}{r_2}.
    \label{eq:intro-eta-delta}
\end{equation}
We also introduce
\begin{equation}
    \varepsilon
    :=
    \sup_{C\cup\underline C}
    \frac{Q^2}{r^2},
    \qquad
    0\leq\varepsilon<1.
    \label{eq:intro-epsilon}
\end{equation}

\noindent For \(0<\omega<2/3\), define
\begin{equation}
\begin{split}
    G_\omega(x)
    &:=
    \frac{1+\frac{\omega}{2}}
         {1-\frac{\omega}{2}}
    \frac{1}{(1+x)^2}
    \Bigg[
       \left(
          \frac{2^{1-\frac{\omega}{2}}}{\omega}
          +
          \frac{1}
          {2^{1+\frac{\omega}{2}}
           \left(1+\frac{\omega}{2}\right)}
       \right)
       x^{1-\frac{\omega}{2}}
       -
       \frac{2}{\omega}x
       -
       \frac{x^2}{1+\frac{\omega}{2}}
    \Bigg].
\end{split}
\label{eq:intro-Gomega}
\end{equation}


\noindent We may now state a rough version of the principal theorem of our paper. A fully rigorous statement will be given later on in Theorem \ref{maintechnical}.

\begin{theorem}[Trapped surface formation]
\label{thm:intro-main}
Let smooth characteristic initial data for the spherically symmetric,
purely magnetic \(SU(2)\) Einstein--Yang--Mills system be prescribed on
\(C\cup\underline C\). Assume that the data satisfy the regularity and no-initial-trapping hypotheses, the \textit{magnetic
subextremality} condition

\[m(u,v_1) -\lvert Q\rvert(u,v_1) \geq 0, \hspace{2mm} \text{for all} \hsp \hsp u: (u,v_1)\in \Cbar \]
 \noindent and
\[
    \varepsilon<1.
\]
Fix \(0<\omega<2/3\). There exists a quantitative short-segment constant
\(\Delta_*>0\), depending only on $\omega, \varepsilon$ and $\sup_{\Cbar} \frac{w^2}{r}$, such that the following holds
whenever, in the fixed double-null normalization,
\[
    0<v_2-v_1\leq\Delta_*.
\]

\noindent Suppose that
\begin{equation}
\begin{split}
    \eta_0
    >
    \max\Bigg\{
       &\frac{13\varepsilon}{\omega}
        +G_\omega(\delta_0),
        \\
       &\frac{9}
       {2^{1+\frac{\omega}{2}}
        (1+\delta_0)^2}
        \delta_0^{1-\frac{\omega}{2}}
        +G_\omega(\delta_0)
    \Bigg\},
\end{split}
\label{eq:intro-main-threshold}
\end{equation}where \[ \eta_0 := \frac{m(u_0,v_2)-m(u_0,v_1)}{r(u_0,v_2)}, \hspace{3mm} \delta_0 := \frac{r(u_0,v_2)-r(u_0,v_1)}{r(u_0,v_2)}.  \]
Let \(u_*\) be the first value of \(u\), within the regular development,
such that
\begin{equation}
    r(u_*,v_2)
    =
    \frac{3\delta_0}{1+\delta_0}\,r_2.
    \label{eq:intro-u-star}
\end{equation}
Then the characteristic development contains a trapped symmetry sphere
in
\[
    [u_0,u_*]\times[v_1,v_2].
\]
More precisely, there exists
\((u_{\mathrm{tr}},v_{\mathrm{tr}})\) in this rectangle such that
\[
    \partial_u r(u_{\mathrm{tr}},v_{\mathrm{tr}})<0,
    \qquad
    \partial_v r(u_{\mathrm{tr}},v_{\mathrm{tr}})<0.
\]
\end{theorem}

\noindent A few words on the intuition behind the bound \eqref{eq:intro-main-threshold}: This lower bound separates the contribution
of the incoming magnetic field from the mass concentration carried by the
outgoing initial segment. The term \(13\varepsilon/\omega\) measures the
loss caused by the non-Abelian magnetic potential, whereas the remaining
terms describe the geometric cost of concentrating a definite amount of
Hawking mass inside a short relative radial interval. In the regime
\(\varepsilon\to0\), the criterion reduces to the corresponding
uncharged concentration mechanism.

\subsection{Difficulties and strategy of the proof}

The proof is not a direct adaptation of the trapped-surface arguments for
the Einstein--scalar field system. There are three related obstructions.

First, the magnetic contribution has two competing effects. From
\eqref{eq:intro-mass-v}, both the kinetic energy
\((\partial_vw)^2\) and the magnetic potential \(Q^2/r^2\) increase the
Hawking mass along \(C_{u_0}\). However, the same potential appears with a
defocusing sign in \eqref{eq:intro-cross-focusing}. A lower bound for the
mass increment alone therefore does not imply a corresponding lower bound
for the focusing strength. One must separate the genuinely focusing part
of the mass from the portion that can be offset by the magnetic potential.

Secondly, \(Q\) is dynamical. In the Einstein--Maxwell system the electric
charge is conserved and its contribution may be treated as a fixed
parameter. Here
\[
    Q=w^2-1,
    \qquad
    \partial_aQ=2w\,\partial_aw,
\]
and the magnetic potential is coupled directly to the radiative component
of the Yang--Mills field through
\eqref{eq:intro-YM-wave}. Hence the analogue of a charge-to-mass ratio must
be estimated throughout the spacetime region rather than merely evaluated
on the initial hypersurfaces.

Third, the incoming initial hypersurface is not assumed to be Minkowskian.
The quantities \(m\), \(Q\), \(\Omega\), \(h\), and \(\underline h\) may
all be nontrivial on \(\underline C\). The proof must therefore propagate
the initial subextremality condition, control the deformation of the
radial width, and quantify the accumulated magnetic error without using
an exact background solution. We now give an outline of the argument:

\vspace{3mm}
\noindent Let us assume, by way of contradiction, that $\mathcal{D}(0,v_1)\cup \mathcal{R}$ does not have a trapped surface. Recall that, initially, $\eta_0<1$. The idea of the proof is to impose a strong enough lower bound on $\frac{\text{d}\eta}{\text{d}u}$, with the ultimate goal of showing that $\eta(u_*)>1$. Once this bound is obtained, it follows that $\frac{2m_2}{r_2}(u_*)\geq \frac{2(m_2-m_1)}{r_2}(u_*)= \eta(u_*)>1$, whence $r_2(u_*)<2m_2(u_*)$ and hence $S(u_*,v_2)$ is a trapped sphere. Given that $(u_*,v_2)\in \mathcal{R}$, we obtain the desired contradiction.

\vspace{3mm}

\noindent In more specific terms, we aim to control $\frac{\text{d}\eta}{\text{d}x},$ where $x= \frac{r_2(u)}{r_2(u_0)}$ (this is equivalent to controlling $\frac{\text{d}\eta}{\text{d}u}$). We calculate:

\begin{align} \label{detadx} \frac{\text{d}\eta}{\text{d}x} =& \hsp \hsp \frac{\frac{\text{d}\eta}{\text{d}u}}{\frac{\text{d}x}{\text{d}u}}= \frac{r_2(u_0)}{\hba_2}\big(-\frac{2\hsp \hba_2}{r_2^2}(m_2-m_1) + \frac{2}{r_2}\partial_u (m_2-m_1)\big) \notag \\ =& -\frac{\eta}{x}+\frac{2}{x \hsp \hba_2} \bigg(\frac{\hba_2 \hsp Q_2^2}{2r_2^2}-\hba_2\hsp (\partial_u w_2)^2 - \frac{\hba_1 \hsp Q_1^2}{2\hsp r_1^2}+h_1(\partial_u w_1)^2\bigg) \notag \\ \leq& -\frac{\eta}{x}-\frac{2\hsp h_2}{x \hsp \hba_2}\hsp \big((\partial_u w_2)^2 - \frac{h_1}{h_2}(\partial_u w_1)^2\big) + \frac{Q_2^2}{x\hsp r_2^2}. \end{align} 

Assuming that the principal part is the leading part here, we expect a naive (heuristic) bound of the form \[ \frac{\text{d}\eta}{\text{d}x}\lesssim -\frac{\eta}{x}.    \]What we will be able to show is that there exist functions $g(x),f(x)$ that are $o(\frac{1}{x})$, such that \be \label{etafg}  \frac{\text{d}\eta}{\text{d}x}+\eta \frac{g(x)}{x}-\frac{f(x)}{x}\leq 0.    \ee This will give a strong enough bound on $\frac{\text{d}\eta}{\text{d}u}$ to infer the existence of a trapped surface. Let us give an overview of how the matter terms are controlled in \eqref{detadx}.

\begin{itemize}
\item The term $\frac{Q_2^2}{x r_2^2}$ is shown to be small enough compared to $\eta$, so that the entire term is absorbed in $\frac{\eta \hsp  g(x)}{x}$ (what $g$ signifies is explained above). 

\item Having obtained the above, we further show that \[(\partial_u w_2 - \partial_u w_1)^2 \leq \big(1 +\frac{\omega}{2}\big) \frac{- \partial_u r_2}{\partial_v r_2}(m_2-m_1)\bigg( \frac{1}{r_1}- \frac{1}{r_2}\bigg)(u),    \]for all $u \in [u_0, u_*]$ and that 

\[  \frac{\Omega_2^{-2}\partial_v r_2}{\Omega_1^{-2}\partial_v r_1}(u) \lesssim e^{-\eta(u)}.   \]What is important to emphasise is that the right-hand side terms in \eqref{detadx} involving matter are controlled above, through those bounds, in terms of the mass difference ratio $\eta$. Putting these together in \eqref{detadx} we arrive at the bound \eqref{etafg}, which in turn gives the desired contradiction as explained.
\end{itemize}

\vspace{3mm}
\noindent



\subsection{Relation to previous work}

The mathematical study of dynamical gravitational collapse in spherical
symmetry was initiated by Christodoulou for the Einstein equations coupled
to a massless scalar field. In a sequence of works
\cite{C91,C93,C94,C99}, he established trapped-surface formation,
analyzed the structure and instability of naked singularities, and proved
weak cosmic censorship for generic asymptotically flat data in the
spherically symmetric scalar-field model. The scalar field contributes to
the Raychaudhuri equations entirely through terms having a favorable
focusing sign.

Charged collapse introduces a competing repulsive mechanism. Quantitative
trapped-surface criteria for charged spherical systems were obtained, for
example, in the Einstein--Maxwell--scalar field setting in \cite{An}.
In the non-abelian problem, the symmetry-free work \cite{A-M-Y} proved semi-global
existence and trapped-surface formation for the Einstein--Yang--Mills
system by developing a gauge-invariant hierarchy of estimates for the
Yang--Mills curvature in a double-null foliation (see also \cite{A19,A17,AnThesis,AnAth,AthL, A-M-Y1, chen, LKR, M-Y1,  Kl-Rod} for trapped surface and MOTS formation results in pure vacuum and with matter models). In that construction,
the incoming initial hypersurface is taken to carry trivial geometric and
Yang--Mills data. The present analysis addresses a complementary problem:
it exploits spherical symmetry to obtain a sharp, explicitly computable
mass-concentration criterion while permitting nontrivial incoming data.

The local and continuation theories for the spherically symmetric system were
recently studied in \cite{Li_Wang}. A distinctive analytic feature of the
purely magnetic reduction is the singular nonlinear potential in
\eqref{eq:intro-YM-wave}, which prevents a direct application of the
\(L^\infty\)-based arguments used for the massless scalar field. The
\(H^1\) extension principle of \cite{Li_Wang} provides an appropriate
continuation framework for the broader program of analyzing weak cosmic
censorship for the spherically symmetric Einstein--Yang--Mills system.

Theorem \ref{thm:intro-main} supplies the collapse component of this
program: sufficiently concentrated Hawking mass produces a trapped
surface despite the dynamically evolving magnetic repulsion. It thereby
connects the classical existence theory of static colored black holes with
a quantitative mechanism by which trapped regions arise from regular
characteristic data.

\vspace{3mm}

\noindent The following is our Main Theorem of the paper:

\begin{theorem} \label{maintechnical}
Let us denote the initial incoming null hypersurface $v=v_1$ by $\Cbar$ and the outgoing null hypersurface $u=u_0$ by $C$.  Define 

\[  \varepsilon:= \sup_{C\cup \Cbar}\frac{Q^2}{r^2} <1.  \] Let $\omega$ be any positive constant in $(0,\frac{2}{3}).$ Choose $v_2-v_1$ sufficiently small such that 

\be \label{firstone}  \frac{64}{(1-\varepsilon)^2}(v_2-v_1)^2 +16(v_2-v_1)L \leq \frac{\omega}{4},  \ee 
and

\be \frac{3 (v_2-v_1)^2}{(1-\varepsilon)^2}+ 2\frac{v_2-v_1}{1-\varepsilon}L\leq \frac{3}{2}\omega, \label{second} \ee

\noindent where $L:= \sup_{(u,v) \in \Cbar} \frac{w^2(u,v)}{r(u,v)}$. Moreover, assume that the initial data along $\Cbar$ are \textit{magnetically subextremal}, i.e. 

\be \label{nonmag} m(u,v_1) \geq \lvert Q \rvert(u,v_1),    \ee for all $u$ such that $(u,v_1)\in \Cbar$. Denote \[ g_{\omega}(x) := \frac{1+\frac{\omega}{2}}{1-\frac{\omega}{2}} \frac{1}{(1+x)^2}\bigg( \big(\frac{2^{1-\frac{\omega}{2}}}{\omega}+\frac{1}{2^{1+\frac{\omega}{2}}(1+\frac{\omega}{2})}\big) x^{1-\frac{\omega}{2}}-\frac{2}{\omega}x - \frac{1}{1+\frac{\omega}{2}}x^2  \bigg).  \]Assume that the following lower bound on $\eta_0$ holds:

\be \eta_0 > \max \bigg\{ \frac{13\varepsilon}{\omega}+g_{\omega}(\delta_0), \hsp \frac{9}{2^{1+\frac{\omega}{2}}(1+\delta_0)^2}\delta_0^{1-\frac{\omega}{2}}+g_{\omega}(\delta_0)\bigg\},\ee where \[ \eta_0 := \frac{m(u_0,v_2)-m(u_0,v_1)}{r(u_0,v_2)}, \hspace{3mm} \delta_0 := \frac{r(u_0,v_2)-r(u_0,v_1)}{r(u_0,v_2)}.  \]We assume $0<\delta_0< \frac{1}{2}$. Then, a trapped surface is guaranteed to form in $[u_0, u_*]\times [v_1, v_2]\subset \mathcal{R}$. Here $u_*$ denotes the minimum value of $u$ for which $r(u,v_2)= \frac{3\delta_0}{1+\delta_0} \hsp r(u_0,v_2)$. 
\end{theorem}

\section{Proof of the Main Theorem}





This section contains the proof of the main theorem. In Section \ref{sectionQ} we estimate the term $\frac{Q_2^2}{xr_2^2}$ from \eqref{detadx} and bound it above in terms of $\eta, \omega$ and the initial data. To achieve this, we shall require a priori bounds on $\partial_u r$ and a crucial monotonicity formula (which is more difficult to prove than in the massless scalar field case, given that the matter field opposes the establishment of monotonicity, but still similar in spirit to \cite{AnLim}). In Section \ref{sectionw} we explain how to bound the terms involving $h$ and $\partial_u w$ on the right-hand side of \eqref{detadx} in terms of the Hawking mass and $r$. Together with the results of Section \ref{sectionQ}, this will ultimately allow us to control the entire right-hand side of \eqref{detadx} in terms of $\eta$, as discussed. 

\subsection{Estimates on $Q$,$r$} \label{sectionQ} In this section we will bound $Q$ in terms of $\eta, \omega$ and the initial data (we note here that $\varepsilon$, which will also be involved in the bounds, is an initial-data quantity). \noindent We begin our proof with a series of lemmata:

\begin{lemma}
Along $v=v_1$, the magnetic subextremality condition \eqref{nonmag} implies \be \notag \frac{m}{r}(u,v_1)\geq \frac{Q^2}{r^2}(u,v_1). \ee
\end{lemma}

\begin{proof}
   Recall that we are assuming the no-trapped-surface condition on $\Cbar$. Consequently, for all $u: (u,v_1) \in \Cbar$, there holds

   \[  \frac{2m}{r}(u,v_1)\leq 1.    \]Employing the magnetic subextremality condition, we have \[ \frac{\lvert Q \rvert}{r}(u,v_1)\leq \frac{m}{r}(u,v_1)\leq \frac{1}{2}.   \]Therefore, \[ \big( \frac{m}{r}-\frac{Q^2}{r^2}\big)(u,v_1) \geq \frac{\lvert Q\rvert}{r}(u,v_1) - \frac{Q^2
   }{r^2}(u,v_1) \geq \frac{\lvert Q \rvert}{r}\big( 1 - \frac{\lvert Q \rvert}{r}\big)(u,v_1) \geq 0.    \]In the last inequality we used $\frac{\lvert Q \rvert}{r}\leq \frac{1}{2}$.
\end{proof}

\begin{lemma}
There holds $\partial_u r \leq -\frac{1-\varepsilon}{2}\Omega^2$, everywhere in $\mathcal{D}(0,v_1)\cup\big( [u_0,0]\times [v_1,\infty]\big)$.
\end{lemma}

\begin{proof}
Given equation \eqref{eq:intro-cross-focusing}, we have \[  \partial_{v}\hba =- \Omega^2 \big(\frac{2m}{r^2}-\frac{Q^2}{r^3}\big).   \] There holds

\begin{align} \partial_{v}(r \hba) = \partial_{v}(r \partial_u r) \notag =& -\Omega^2\big(\frac{2m}{r^2}-\frac{Q^2}{r^3}\big) r + \Omega^2 h \hsp \hba \\ =& \Omega^2 \big(-\frac{2m}{r}+\frac{Q^2}{r^2}\big) + \Omega^2 \big(\frac{2m}{r}-1\big) = -\Omega^2 \big(1-\frac{Q^2}{r^2}\big).   \end{align} Using the assumption that $\varepsilon <1 $ and $\Omega^2 =\frac{1}{4}$ on $C$, we arrive at \be -\frac{1}{4}\leq \partial_{v}(r \partial_u r)(u_0,v) < -\frac{1}{4}+\frac{\varepsilon}{4}, \ee for all $v: (u_0,v) \in C$. An integration on both sides yields

\be -\frac{v}{4r(u_0,v)} \leq \partial_u r(u_0,v)\leq \frac{-(1-\varepsilon)v}{4r(u_0,v)}.  \label{useful}      \ee For the first inequality in \eqref{useful}, evaluating at $v=0$ and using  gives

\be -\frac{1}{4 \partial_{v}r(u_0,0)}\leq \partial_u r(u_0,0)= -\partial_{v}r(u_0,0) \Rightarrow \partial_{\ubar}r(u_0,0)\leq \frac{1}{2}. \ee Taking into account that $\Omega^2=\frac{1}{4}$ on $C$, equation \eqref{eq:intro-Raychaudhuri-v} implies \be \partial_{v}\partial_{v} r(u_0, \cdot) \leq 0,\ee meaning that $r(u_0,\cdot)$ is concave. This gives that

\be  \frac{r}{v}(u_0,v) \leq \partial_{v}r(u_0,0). \ee Here we have used the fact that $r(u_0,0)= v(u_0,0)=0$.  Substituting into the second inequality of \eqref{useful} yields 
 \be \partial_u r(u_0,v) \leq -\frac{1-\varepsilon}{2}. \ee Finally, by equation \eqref{eq:intro-Raychaudhuri-u}, the quantity $\Omega^{-2}\partial_u r$ is decreasing along incoming null geodesics. Hence for a general point in $\mathcal{D}(0,v_1)\cup\big([u_0,0]\times [v_1,\infty)\big)$, we have \[ \Omega^{-2}\partial_u r \leq - \frac{1-\varepsilon}{2}.   \]

\end{proof}

\begin{proposition}\label{propositionforQ}
Fix $0< \omega < \frac{2}{3}.$ Choose $v_2-v_1$ sufficiently small satisfying \eqref{conditionv}. Let $v_1< v_a \leq v_2$. Assume that $\frac{r_2(u)}{r_1(u)}\leq \frac{3}{2}$ , i.e. that $\delta(u) \leq \frac{1}{2}$ for $u \in [u_0,0]$ and that $\mathcal{R} := [u_0,0]\times [v_1,v_2]$ is free of trapped surfaces. Then the following inequality holds:

\[ \frac{Q_a^2(u)}{r_a^2(u)} \leq \frac{\omega}{4}\eta_a(u)+\frac{2 Q_1^2(u)}{r_1^2(u)}.     \]Here,  \[ \eta_a(u):= \frac{2\big(m_a(u)-m_1(u)\big)}{r_a} .  \]The subscript a indicates a quantity evaluated at the point $(u,v_a)$.
\end{proposition}

\begin{proof}
We write

\begin{align} \notag
Q_a^2 =&\bigg(\int_{v_1}^{v_a}\partial_{v}Q \hspace{1mm} \text{d}v +Q_1\bigg)^2 \leq 2 \bigg(\inta \partial_v Q \hspace{1mm} \text{d}v\bigg)^2 +2Q_1^2 \\ \leq& 2\bigg(\inta 2 w \hsp \partial_v w \dv \bigg)^2 + 2Q_1^2  \leq 8  \inta w^2 \dv \cdot \inta (\partial_v w)^2 \dv + 2 Q_1^2.   \label{Qbound}
\end{align}The second integral can be bounded as follows:

\begin{align}
\inta (\partial_v w)^2 \dv =& \inta \frac{-\Omega^{-2}\hb (\partial_v w)^2}{-\Omega^{-2}\hb} \dv \notag \\ \leq& \frac{2}{1-\varepsilon}\inta \partial_v m - \frac{\Omega^2 h Q^2}{2r^2} \leq \frac{2(m_a-m_1)}{1-\varepsilon}.  \label{di}
\end{align}

For the first integral, we have

\begin{align}
\inta w^2 \dv = &\inta\big(\inta \partial_{v^{\prime}}w\hspace{1mm} \text{d}v^{\prime} +w_1\big)^2 \dv \notag \\ \leq& 2 \inta \bigg[ \big(\inta \partial_{v^{\prime}}w \dvprime\big)^2 +w_1^2 \bigg] \dv \notag \\ \leq& 2 \inta \bigg[ (v-v_1) \inta (\partial_{v^{\prime}}w)^2 \dvprime +w_1^2 \bigg] \dv \notag\\ \leq& 2 \inta \bigg[ (v-v_1) \inta \frac{-\Omega^{-2}\hb (\partial_{v^{\prime}}w)^2}{-\Omega^{-2}\hb} \dvprime +w_1^2 \bigg] \dv \notag\\ \leq& 2(v_a-v_1) \cdot \frac{2}{1-\varepsilon}\inta \int_{v_1}^v \partial_{v^{\prime}}m \hspace{1mm} \text{d}v^{\prime}\dv +2(v_a-v_1) \inta w_1^2 \dv \notag \\ \leq& \frac{4}{1-\varepsilon}(v_a-v_1)^2 (m_a-m_1) +2(v_a-v_1) w_1^2. \label{fi}
\end{align}
Therefore, putting \eqref{Qbound}-\eqref{fi} together, we have 

\begin{align}
&Q_a^2- 2Q_1^2 \notag \\ \leq& 8\cdot \frac{2(m_a-m_1)}{1-\varepsilon}\cdot \big( \frac{4}{1-\varepsilon}(v_a-v_1)^2(m_a-m_1)^2 + 16 (v_a-v_1)(m_a-m_1)w_1^2\big).
\end{align}
Dividing by $r_a^2$, we have:

\begin{align}
\frac{Q_a^2}{r_a^2} \leq& \frac{64}{(1-\varepsilon)^2}(v_a-v_1)^2 \eta_a^2 +16(v_a-v_1)\eta_a \frac{w_1^2}{r_a} + \frac{2Q_1^2}{r_1^2} \\ \leq& \bigg(\frac{64}{(1-\varepsilon)^2}(v_a-v_1)^2 +16(v_a-v_1)\frac{w_1^2}{r_1}\bigg)\eta_a + \frac{2Q_1^2}{r_1^2}. 
\end{align}Here we have used the fact that $\eta_a = \frac{2(m_a-m_1)}{r_a} \leq \frac{2m_a}{r_a}\leq 1$. By denoting $L:=\sup_{\Cbar} \frac{w_1^2}{r_1}$ and choosing $v_2-v_1$ sufficiently small such that 

\be \label{conditionv}  \frac{64}{(1-\varepsilon)^2}(v_a-v_1)^2 +16(v_a-v_1)L \leq \frac{\omega}{4},       \ee we arrive at the desired result.

\end{proof}
\noindent We are ready to obtain a key monotonicity formula, in the form of the following Proposition:

\begin{proposition} \label{propmonotonicity}
Assume that the initial data along $\Cbar$ is magnetically subextremal and that $\mathcal{R}$ is free of trapped surfaces. Then $\partial_u\partial_v r \leq 0$ in $[u_0, u_*] \times [v_1,v_2]$ and $\delta(u):= \frac{r_2(u)}{r_1(u)}-1 \leq \frac{1}{2}$ for $u \in [u_0, u_*],$ where $u_*$ is defined so that $x(u_*)= \frac{3\delta_0}{1+\delta_0}.$ 
\end{proposition}

\begin{proof}
Let $x^{\prime} =: \inf \{ x \in [\frac{3\delta_0}{1+\delta_0},1] \mid \delta(y)\leq \hsp  \frac{1}{2}\hsp \text{holds} \hsp \text{for}\hspace{1mm} y\in [x,1]\}.$ We will aim to show that $x^{\prime} = \frac{3\delta_0}{1+\delta_0}.$ This proves the claim that $\delta(u)\leq \frac{1}{2}$ for all $u\in [u_0,u_*]$.

\vspace{3mm}

\noindent Proposition \eqref{propositionforQ} is applicable, since $\frac{r_2(x^{\prime})}{r_1(x^{\prime})}\leq \frac{3}{2}$. Thus, for all $x \in [x^{\prime},1]$ and $v_a \in [v_1,v_2]$, we have

\begin{align} \notag 
\frac{Q^2}{r^2}(x,v_a) \leq& \frac{\omega}{4}\eta_a + \frac{2Q_1^2}{r^2} \leq \eta_a + \frac{2Q_1^2}{r^2} \\ =& \frac{2m_a}{r}- \frac{2}{r}\big(m_1-\frac{Q_1^2}{r}\big) \notag \\ \leq& \frac{2m_a}{r}- \frac{2}{r}\big(m_1-\frac{Q_1^2}{r_1}\big)
\end{align}However, the magnetic sub-extremality assumption  \eqref{nonmag} implies $m_1- \frac{Q_1^2}{r_1} \geq 0$. Hence

\[ \frac{Q^2}{r^2}(x,v_a) \leq \frac{2m}{r}(x,v_a). \]
We now have 

\begin{align}
\partial_u \partial_v r = -\frac{\Omega^2}{r} \big(\frac{2m}{r}-\frac{Q^2}{r^2}\big).
\end{align}We thus have $\partial_u \partial_v r\leq 0$ in the region $[u_0, u^{\prime}]\times [v_1,v_2]$, where $u^{\prime}$ is defined by $x(u^{\prime})=x^{\prime}$. Integrating with respect to $u$ yields \[ \partial_v r(u)\leq \partial_v r(u_0).   \]Integrating this inequality with respect to $v$, we obtain 

\begin{align}
r_2(u)-r_1(u)\leq r_2(u_0)-r_1(u_0),
\end{align}for all $u \in [u_0,u^{\prime}]$. We can use this information to derive a bound for $\delta$:

\begin{align}\notag 
\delta(u)=& \frac{r_2(u)}{r_1(u)}-1 = \frac{r_2(u)-r_1(u)}{r_2(u)-(r_2(u)-r_1(u))} \leq \frac{r_2(u_0)-r_1(u_0)}{r_2(u)-(r_2(u_0)-r_1(u_0))} \notag \\ \leq& \frac{\delta_0}{\frac{r_2(u_0)}{r_1(u_0)}-\delta_0} = \frac{\delta_0}{x(u)(1+\delta_0) -\delta_0},
\end{align}for all $u \in [u_0, u^{\prime}]$. If $x^{\prime}>\frac{3 \delta_0}{1+\delta_0}$, we have $\delta(x^{\prime})< \frac{\delta_0}{3\delta_0-\delta_0}=\frac{1}{2}.$ By the continuity of $\delta$, there exists some $x^{\prime \prime}< x^{\prime}$ such that $\delta(x)< \frac{1}{2}$ for all $x \in [x^{\prime \prime},1]$, which is a contradiction to the infimum property of $x^{\prime}$. Therefore, there must hold $x^{\prime} = \frac{3\delta_0}{1+\delta_0}$.
\end{proof}

\noindent The above lead us to the following lemma:

\begin{lemma} \label{lemmaomegaeta}
Assume that the initial data along $\Cbar$ is magnetically subextremal and that $\mathcal{R}$ is free of trapped surfaces. Then for all $(u,v_a) \in [u_0, u_*] \times [v_1,v_2]$, we have the following estimate for the charge:

\[  \frac{Q_a^2}{r_a^2}\leq \frac{\omega}{4}\eta_a+2\varepsilon,  \]where we recall that $\varepsilon = \sup_{C \cup \Cbar}\frac{Q^2}{r^2} <1$. If, moreover, $\eta_a \geq \frac{8\varepsilon}{\omega},$ then $\frac{Q_a^2}{r_a^2}\leq \frac{\omega}{2}\eta_a$.
\end{lemma}

\begin{proof}
The hypothesis of this lemma satisfies that of Proposition \ref{propmonotonicity}, therefore $\delta(u) \leq \frac{3}{2}$ for all $u \in [u_0, u_*]$. Thus the hypothesis of Proposition \ref{propositionforQ} is satisfied. Applying Proposition \ref{propositionforQ} proves the first part. For the second part, notice that if $\eta_a \geq \frac{8\varepsilon}{\omega}$, we have

\[  \frac{Q_a^2}{r_a^2}\leq \frac{\omega}{4}\eta_a +2\varepsilon\leq \frac{\omega}{4}\eta_a + \frac{\omega}{4}\eta_a = \frac{\omega}{2}\eta_a.    \]
\end{proof}

\subsection{Estimates for $\partial_u w$.}\label{sectionw}
We now move to estimates involving a first derivative of the Yang-Mills parameter $w$. Our main estimate is presented in the form of the lemma below.
\begin{lemma} \label{Thetalemma}
Define $\Theta:= \partial_u w_2 - \partial_u w_1$. Suppose that the initial data along $\Cbar$ is not super-charged and that $\mathcal{D}(0,v_1)\cup \mathcal{R}$ is free of trapped surfaces. If $\eta \geq \frac{8 \varepsilon}{\omega}$, then 

\[  \Theta(u)^2 \leq \big(1 +\frac{\omega}{2}\big) \frac{- \partial_u r_2}{\partial_v r_2}(m_2-m_1)\bigg( \frac{1}{r_1}- \frac{1}{r_2}\bigg)(u),    \]for all $u \in [u_0, u_*]$.
\end{lemma}

\begin{proof}
Recall that $w$ satisfies the equation

\[  \partial_v \partial_u w +  \Omega^2\frac{Q}{r^2}w=0.   \]As such, 

\begin{align}
\Theta^2 =(\partial_u w_2 - \partial_u w_1)^2 = \big(\int_{v_1}^{v_2} - \Omega^2\frac{Q}{r^2}w \dv \big)^2.
\end{align}We now bound 

\begin{align}\notag 
\big(\int_{v_1}^{v_2} - \Omega^2\frac{Q}{r^2}w \dv \big)^2 =& \big( \int_{v_1}^{v_2} \frac{Q}{r^{3/2}\Omega^{-2}} r^{-1/2}w  \dv\big)^2 \\ \leq& \int_{v_1}^{v_2} \frac{Q^2}{r^2}\cdot \frac{1}{r}\hsp \frac{1}{\Omega^{-2}}\frac{1}{\Omega^{-2}} \dv \cdot \int \frac{w^2}{r}\dv \notag \\ \leq& \int_{r_1}^{r_2} \frac{\omega}{2}\eta_a \frac{1}{r}\frac{1}{\Omega^{-2} \partial_v r}\frac{-\partial_u r}{-\Omega^{-2}\partial_u r}\hsp \text{d}r \int_{v_1}^{v_2}\frac{w^2}{r} \dv. \label{thetafirst}
\end{align}For the first integral, given that $\partial_u \partial_v r\leq 0$ from Proposition \ref{propmonotonicity}, we have $\partial_u r_2 \leq \partial_u r$. Moreover, since $\Omega_2^{-2}\partial_v r_2 \leq \Omega^{-2}\partial_v r$ (from equation \eqref{eq:intro-Raychaudhuri-v}), we have

\begin{align} \label{thetatwo}
&\int_{r_1}^{r_2} \frac{\omega}{2}\eta_a\frac{1}{r}\frac{1}{\Omega^{-2}\partial_v r}\frac{-\partial_u r}{-\Omega^{-2}\partial_u r} \hsp \text{d}r\notag \\ \leq& - \frac{2\omega}{1-\varepsilon}\frac{\partial_u r_2}{\Omega_{2}^{-2}\partial_v r_2}(m_2-m_1)\big(\frac{1}{r_1}-\frac{1}{r_2}\big).
\end{align}We moreover have \begin{align}
\inta \frac{w^2}{r} \dv = &\inta \frac{1}{r} \big(\inta \partial_{v^{\prime}}w\hspace{1mm} \text{d}v^{\prime} +w_1\big)^2 \dv \notag \\ \leq& 2 \inta \bigg[ \frac{1}{r}\big(\inta \partial_{v^{\prime}}w \dvprime\big)^2 +\frac{1}{r} w_1^2 \bigg] \dv \notag \\ \leq& \frac{2}{r_1} \inta \bigg[ (v-v_1) \inta (\partial_{v^{\prime}}w)^2 \dvprime +w_1^2 \bigg] \dv \notag\\ \leq& \frac{2}{r_1} \inta \bigg[ (v-v_1) \inta \frac{-\Omega^{-2}\hb (\partial_{v^{\prime}}w)^2}{-\Omega^{-2}\hb} \dvprime +w_1^2 \bigg] \dv \notag\\ \leq& \frac{2}{r_1}(v_2-v_1) \cdot \frac{2}{1-\varepsilon}\inta \int_{v_1}^v \partial_{v^{\prime}}m \hspace{1mm} \text{d}v^{\prime}\dv +\frac{2}{r_1} \inta w_1^2 \dv \notag \\ \leq& \frac{4}{r_1(1-\varepsilon)}(v_2-v_1)^2 (m_2-m_1) +\frac{2}{r_1}(v_2-v_1) w_1^2\notag \\ \leq& \frac{2(v_2-v_1)^2}{1-\varepsilon} \frac{r_2}{r_1}\frac{2 (m_2-m_1)}{r_2} + 2(v_2-v_1) w_1 \notag \\ \leq& \frac{3 (v_2-v_1)^2}{1-\varepsilon}+ 2(v_2-v_1)L, \label{fi2}
\end{align}where we have used the fact that $\frac{r_2}{r_1}\leq \frac{3}{2}$ and the fact that $\eta = \frac{2(m_2-m_1)}{r_2}<1$, because of the assumption that there are no trapped surfaces or MOTS. We thus have:

\begin{align}
\int_{v_1}^{v_2} \frac{w^2}{r} \dv \leq \frac{3 (v_2-v_1)^2}{1-\varepsilon}+ 2(v_2-v_1)L\leq \frac{3}{2}\omega (1-\varepsilon) ,
\end{align}where we have used assumption \eqref{second}. Plugging this back to \eqref{thetafirst} and using \eqref{thetatwo}-\eqref{fi2}, we obtain

\begin{align}
\Theta^2 \leq& \frac{2\omega}{1-\varepsilon}\cdot \frac{3}{2} \omega (1-\varepsilon)\cdot  \frac{-\partial_u r_2}{\Omega_{2}^{-2}\partial_v r_2}(m_2-m_1)\big(\frac{1}{r_1}-\frac{1}{r_2}\big)\notag \\ \leq& 3\hsp \omega^2 \frac{-\partial_u r_2}{\Omega_{2}^{-2}\partial_v r_2}(m_2-m_1)\big(\frac{1}{r_1}-\frac{1}{r_2}\big)\notag \\ \leq& \big(1+\frac{\omega}{2}\big) \frac{-\partial_u r_2}{\Omega_{2}^{-2}\partial_v r_2}(m_2-m_1)\big(\frac{1}{r_1}-\frac{1}{r_2}\big),\notag
\end{align}because $0<\omega<\frac{2}{3}$ implies $3\hsp \omega^2< \frac{\omega}{2}+1$. The result follows.

\end{proof}

\noindent We proceed to prove estimates for $h$:

\begin{lemma} \label{hlemma}
Assume that $\eta \geq \frac{8 \varepsilon}{\omega}$, the initial data along $\Cbar$ is not magnetic-supercharged and that $\mathcal{D}(0,v_1)\cap \mathcal{R}$ is devoid of trapped surfaces. Then

\[ \frac{h_2(u)}{h_1(u)}\leq e^{-(1-\frac{\omega}{2})}\eta(u).      \]
\end{lemma}

\begin{proof}
We recall the equation \[\partial_{v} h = - 2\Omega^{-2}\frac{(\partial_v w)^2}{r}.\] Dividing both sides by $h$ and integrating from $v_1$ to $v_2$, we have

\begin{align}
\ln \lvert h_2 \rvert - \ln \lvert h_1 \rvert = - 2\int_{v_1}^{v_2} \Omega^{-2} \frac{(\partial_v w)^2}{rh}. \end{align}

We have \be \frac{1}{r-2m}\big(\partial_v m - \frac{\Omega^2 Q^2 h}{2r^2}\big) = \Omega^{-2} \frac{(\partial_v w)^2}{r h}.  \ee

Therefore, for any $u \in [u_0, u_*]$, there holds

\begin{align}
\ln\big( \frac{h_2}{h_1}\big) =& -2 \int_{v_1}^{v_2} \frac{1}{r-2m}\big(\partial_v m - \Omega^2 \frac{Q^2 h}{2r^2}\big) \dv \notag \\ \leq& -2 \int_{v_1}^{v_2}\frac{1}{r}\big(\partial_v m - \Omega^2 \frac{ Q^2 h}{2r^2}\big) \dv \leq - \frac{2}{r_2} \int_{v_1}^{v_2} \big( \partial_v m - \frac{Q^2 \partial_v r}{2r^2}\big) \dv \notag \\ \leq& -\eta + \frac{2}{r_2} \int_{r_1}^{r_2}\frac{Q^2}{2r^2} \hspace{.5mm} \text{d}r. 
\end{align}By Lemma \ref{lemmaomegaeta}, we have \[  \frac{Q(u,v)^2}{r(u,v)^2}\leq \frac{\omega}{2} \frac{2(m(u,v)-m(u,v_1)}{r(u,v)}.    \]Therefore, 

\begin{align}
\frac{2}{r_2} \int_{r_1}^{r_2}\frac{Q^2}{2r^2} \hspace{.5mm} \text{d}r \leq& \frac{\omega}{2r_2}\int_{r_1}^{r_2} \frac{2(m(u,v)-m(u,v_1))}{r(u,v)}\hspace{.5mm}\text{d}r \notag \\ \leq& \omega \hsp \frac{m_2-m_1}{r_2} \ln\big(\frac{r_2}{r_1}\big) \leq \frac{\omega}{2} \ln\big(\frac{3}{2}\big) \eta \leq \frac{\omega}{2}\hsp  \eta. 
\end{align}Therefore, \[ \ln\big(\frac{h_2}{h_1}\big) \leq -(1-\frac{\omega}{2})\eta.   \]
\end{proof}

\subsection{Finishing the proof}\label{sec:trapped-surface}

We are now ready to provide a lower bound on $\frac{\text{d}\eta}{\text{d}u}$. As is obvious, Lemmata \ref{lemmaomegaeta}, \ref{Thetalemma} and \ref{hlemma} are only applicable in a region where $\eta \geq \frac{8 \varepsilon}{\omega}$. However, this bound itself is also required to control $\frac{\text{d}\eta}{\text{d}u}$ and thus establish $\eta \geq \frac{8 \varepsilon}{\omega}$. We will therefore employ a bootstrap argument. 

\begin{lemma}
Assume that the region $\mathcal{D}(0,v_1)\cap \mathcal{R}$ is free of trapped surfaces and the initial data along $\Cbar$ is magnetically subextremal. Then, if $\eta_0 \geq \frac{13\varepsilon}{\omega}+g_{\omega}(\delta_0)$, we have $\eta(x) \geq \frac{12\varepsilon}{\omega}$ for all $x= x(u)\in [\frac{3\delta_0}{1+\delta_0},1]$. Here, $g_{\omega}(x)$ is defined as:

\[ g_{\omega}(x):= \frac{1+\frac{\omega}{2}}{1-\frac{\omega}{2}} \hsp \frac{1}{(1+x)^2}\bigg(\bigg( \frac{2^{1-\frac{\omega}{2}}}{\omega}+ \frac{1}{2^{1+\frac{\omega}{2}}(1+\frac{\omega}{2})}\bigg)x^{1-\frac{\omega}{2}}-\frac{2}{\omega}x-\frac{1}{1+\frac{\omega}{2}}x^2\bigg).    \]
\end{lemma}

\begin{proof}
Let $\upr:= \sup\{ u \in [u_0,u_*] \mid \eta(s) \geq \frac{12\varepsilon}{\omega} \hsp \text{for} \hsp \text{all} \hsp s\in [u_0,u]\}.$ The quantity we will estimate is $\frac{\text{d}\eta}{\text{d}x}$, where $x= \frac{r_2(u)}{r_2(u_0)}$. There holds 

\begin{align} \label{coreinternal} \frac{\text{d}\eta}{\text{d}x} =& \frac{\frac{\text{d}\eta}{\text{d}u}}{\frac{\text{d}x}{\text{d}x}}= \frac{r_2(u_0)}{\hba_2}\big(-\frac{2\hsp \hba_2}{r_2^2}(m_2-m_1) + \frac{2}{r_2}\partial_u (m_2-m_1)\big) \notag \\ =& -\frac{\eta}{x}+\frac{2}{x \hsp \hba_2} \bigg(\frac{\hba_2 \hsp Q_2^2}{2r_2^2}-\hba_2\hsp (\partial_u w_2)^2 - \frac{\hba_1 \hsp Q_1^2}{2\hsp r_1^2}+h_1(\partial_u w_1)^2\bigg) \notag \\ \leq& -\frac{\eta}{x}-\frac{2\hsp h_2}{x \hsp \hba_2}\hsp \big((\partial_u w_2)^2 - \frac{h_1}{h_2}(\partial_u w_1)^2\big) + \frac{Q_2^2}{x\hsp r_2^2}. \end{align} 
We now focus our attention on the region $[u_0, \upr]$. Since $\eta \geq \frac{12\varepsilon}{\omega}$ by the bootstrap asumption in $[x^\prime,1]$, there also holds $\eta \geq \frac{8\varepsilon}{\omega}$, whence the conditions of Lemma \ref{hlemma} hold. We thus obtain:

\begin{align}
(\partial_u w_2)^2 - \frac{h_1}{h_2}(\partial_u w_1)^2 \leq& (\partial_u w_2)^2 - e^{\eta(1-\frac{\omega}{2})} (\partial_u w_1)^2\notag \\ =& \Theta^2 +2\Theta \hsp (\partial_u w_1) + \big(1-e^{\eta(1-\frac{\omega}{2})}\big) (\partial_u w_1)^2. \label{intermediate1}
\end{align}The expression above on the right-hand side of \eqref{intermediate1} is then bounded as follows:

\begin{align}
\Theta^2 +2\Theta \hsp (\partial_u w_1) + (1-e^{\eta(1-\frac{\omega}{2})}) (\partial_u w_1)^2 \leq& \big(1 + \frac{1}{e^{\eta(1-\frac{\omega}{2}}-1}\big) \Theta^2 \notag \\ \leq& \big(1 + \frac{1}{\eta(1-\frac{\omega}{2})}\big) \Theta^2,
\end{align}since $\eta(1-\frac{\omega}{2}) \geq 0$. Plugging this bound back to \eqref{coreinternal}, we have

\begin{align}
\frac{\text{d}\eta}{\text{d}x}\leq -\frac{\eta}{x}- \frac{2h_2}{x \underline{h}_2}\big(1 + \frac{1}{\eta(1-\frac{\omega}{2})}\big) \Theta^2 + \frac{Q^2}{xr_2^2}. 
\end{align}

\noindent Applying Lemma \ref{Thetalemma}, we have

\begin{equation}\label{mediumestimate} \frac{\text{d}\eta}{\text{d}x}\leq -\frac{\eta}{x} + \frac{\eta}{x}\big(\frac{\omega}{2}+1\big)\bigg(1+ \frac{1}{\eta(1-\frac{\omega}{2})}\bigg) \big(\frac{r_2}{r_1}-1\big) + \frac{Q_2^2}{xr_2^2}. \end{equation}By employing the monotonicity formula from Proposition 
\ref{propmonotonicity}, we get

\begin{align}\label{boundmonot}
\delta(u)= \frac{r_2(u)-r_1(u)}{r_2(u)-(r_2(u)-r_1(u))}\leq& \frac{r_2(u_0)-r_1(u_0)}{r_2(u_0)-(r_2(u_0)-r_1(u_0))} \notag \\ =& \frac{\delta_0}{x(u)(1+\delta_0)-\delta_0}. 
\end{align}The last term on the right-hand side of \eqref{mediumestimate} will be bounded through the use of Lemma \ref{lemmaomegaeta}. Looking at \eqref{mediumestimate} and using \eqref{boundmonot} together with Lemma \ref{lemmaomegaeta} (which is applicable since $\eta \geq \frac{12\varepsilon}{\omega}>\frac{8\varepsilon}{\omega}$ in our region of interest),

\begin{align}
\frac{\text{d}\eta}{\text{d}x}\leq& \eta \bigg(\big(1+\frac{\omega}{2}\big) \frac{\delta_0}{x^2(1+\delta_0)-x\delta_0)}-\frac{1}{x}\bigg) +\frac{1+\frac{\omega}{2}}{1-\frac{\omega}{2}} \frac{1}{x}\frac{\delta_0}{x(1+\delta_0)-\delta_0}+\frac{Q_2^2}{xr_2^2} \notag \\ \leq&\eta \bigg(\big(1+\frac{\omega}{2}\big) \frac{\delta_0}{x^2(1+\delta_0)-x\delta_0)}-\frac{1}{x}\bigg) +\frac{1+\frac{\omega}{2}}{1-\frac{\omega}{2}} \frac{1}{x}\frac{\delta_0}{x(1+\delta_0)-\delta_0}+\frac{\omega}{2}\frac{\eta}{x} \notag \\ =& -\frac{\eta}{x} \bigg(1-\frac{\omega}{2} - (1+\frac{\omega}{2}) \frac{\delta_0}{x(1+\delta_0)-\delta_0} \bigg) +\frac{1+\frac{\omega}{2}}{1-\frac{\omega}{2}} \frac{1}{x}\frac{\delta_0}{x(1+\delta_0)-\delta_0}.\notag
\end{align}It thus follows that if we define

\[ g(x):=   1-\frac{\omega}{2} - (1+\frac{\omega}{2}) \frac{\delta_0}{x(1+\delta_0)-\delta_0}  \]and \[f(x) :=\frac{1+\frac{\omega}{2}}{1-\frac{\omega}{2}} \frac{1}{x}\frac{\delta_0}{x(1+\delta_0)-\delta_0},  \]we obtain the following differential inequality for all $x\in [x^{\prime},1]$:

\[ \frac{\text{d}\eta}{\text{d}x}+\eta \frac{g(x)}{x}-\frac{f(x)}{x}\leq 0.    \]A straightforward calculation gives, solving the ODE, the bound

\be \eta_0 - e^{-G(x)}\eta(x)\leq F(x), \label{almostlast}\ee

where \be G(x) := \int_x^1 \frac{g(t)}{t} \hsp \text{d}t= \ln \bigg(\frac{(x(1+\delta_0)-\delta_0)^{1+\frac{\omega}{2}}}{x^2}\bigg)\ee and

\begin{align}
&F(x) := \int_x^1 e^{-G(t)}\frac{f(t)}{t} \text{d}t \notag \\ =& \frac{1+\frac{\omega}{2}}{1-\frac{\omega}{2}} \hsp \frac{\delta_0}{(1+\delta_0)^2}\bigg(\frac{2}{\omega}\hsp \frac{1}{(x(1+\delta_0)-\delta_0)^{\frac{\omega}{2}}}-\frac{2}{\omega}+\frac{1}{1+\frac{\omega}{2}}\hsp\frac{\delta_0}{(x(1+\delta_0)-\delta_0)^{\frac{\omega}{2}}}- \frac{\delta_0}{1+\frac{\omega}{2}}\bigg). 
\end{align}Since $F$ is monotonically decreasing, it attains its maximum in the interval $\big[ \frac{3\delta_0}{1+\delta_0},1\big]$ at $x= \frac{3 \delta_0}{1+\delta_0}$. Therefore, 

\begin{align}
F(x)\leq F\big(\frac{3\delta_0}{1+\delta_0}\big)= g_{\omega}(\delta_0).
\end{align}Substituting into \eqref{almostlast}, we get

\begin{align} \label{last}\notag \eta(x)\geq e^{G(x)}(\eta_0-F(x))=& \frac{(x(1+\delta_0)-\delta_0)^{1+\frac{\omega}{2}}}{x^2}(\eta_0-F(x)) \\ \geq& \frac{(x(1+\delta_0)-\delta_0)^{1+\frac{\omega}{2}}}{x^2}(\eta_0-g_{\omega}(\delta_0)). \end{align}Since $\omega<\frac{2}{3}$, it is a straightforward calculation that 

\[  \sup_{x \in[x_*,1]} \frac{x^2}{\big(x(1+\delta_0)-\delta_0\big)^{1+\frac{\omega}{2}}}=1.   \]Taking into account the hypothesis 

\[ \eta_0 > \frac{13 \varepsilon}{\omega}+g_{\omega}(\delta_0),   \]we obtain the inequality

\[  \eta_0 > \frac{13\varepsilon}{\omega}+g_{\omega}(\delta_0) \geq \frac{13 \varepsilon}{\omega}  \frac{x^2}{\big(x(1+\delta_0)-\delta_0\big)^{1+\frac{\omega}{2}}} +g_{\omega}(\delta_0), \]for $x\in[x^{\prime},1]$. Substituting the above into \eqref{last}, we obtain \[ \eta(x) \geq \frac{13\varepsilon}{\omega}, \hspace{2mm} \text{for all} \hspace{1mm} x\in[x^{\prime},1].    \]However, by continuity of $\eta,$ we can find $\tilde{x}<x^{\prime}$ such that for all $x \in[\tilde{x},1]$ there holds $\eta(x)\geq \frac{12\varepsilon}{\omega}$, or equivalently $\eta(u)\geq \frac{12\varepsilon}{\omega}$ for all $u \in [u_0,\tilde{u}]$, contradicting the supremum property of $u^{\prime}$.
\end{proof}
\noindent We finally prove Theorem \ref{maintechnical}.

\begin{proof}
Assume by contradiction that $\mathcal{R}$ contains no trapped surfaces or MOTS, equivalently that $\partial_v r_2 >0$ for all $u \in [u_0,u_{*}]$. Then Lemma \ref{Thetalemma} applies and \eqref{almostlast} holds for $x \in [x_*,1]$. This equation rewrites as 
 \be \eta_0 \leq e^{-\eta(x)}G(x)+F(x) < e^{-G(x)}+F(x).\ee
By setting $x=x_*= \frac{3 \delta_0}{1+\delta_0},$ we get, via a straightforward calculation, 

\be e^{-G(x_*)}+F(x_*) < \frac{9}{2^{1+\frac{\omega}{2}}(1+\delta_0)^2} \delta_0^{1-\frac{\omega}{2}} +g_{\omega}(\delta_0), \ee
hence $\eta_0 <\frac{9}{2^{1+\frac{\omega}{2}}(1+\delta_0)^2} \delta_0^{1-\frac{\omega}{2}} +g_{\omega}(\delta_0)$. This is a contradiction and completes the desired proof.
\end{proof}

\end{document}